\documentclass[%
 aip,
 jmp,
 amsmath,amssymb, 
reprint, onecolumn 
]{revtex4-1}

\usepackage{graphicx}
\usepackage{dcolumn}
\usepackage{bm}

\usepackage[utf8]{inputenc}
\usepackage[T1]{fontenc}
\usepackage{mathptmx}
\usepackage{etoolbox}

\makeatletter
\def\@email#1#2{%
 \endgroup
 \patchcmd{\titleblock@produce}
  {\frontmatter@RRAPformat}
  {\frontmatter@RRAPformat{\produce@RRAP{*#1\href{mailto:#2}{#2}}}\frontmatter@RRAPformat}
  {}{}
}%
\makeatother

\usepackage[pdftex]{color} 
\definecolor{bleu_sombre}{rgb}{0,0,0.6} 
\definecolor{bs}{rgb}{0,0,0.6}  
\definecolor{rouge_sombre}{rgb}{0.8,0,0}
\definecolor{rs}{rgb}{0.8,0,0}
\definecolor{vert_sombre}{rgb}{0,0.6,0}
\definecolor{vs}{rgb}{0,0.6,0}
\usepackage[plainpages=false,colorlinks,linkcolor=bleu_sombre,
citecolor=rouge_sombre,urlcolor=vert_sombre,breaklinks]{hyperref}

\usepackage{bm}
\usepackage{enumerate}

\usepackage{amsthm}
\theoremstyle{plain} 
\newtheorem{theorem}{Theorem}[section]
\newtheorem{lemma}[theorem]{Lemma}

\newtheorem{assumption}{Assumption}
\newtheorem{property}[theorem]{Property}

\theoremstyle{definition}
\newtheorem{remark}[theorem]{Remark}

\newtheorem{definition}[theorem]{Definition}

\newcommand{\iu}{{\rm i}}

\renewcommand{\leq}{\leqslant}	
\renewcommand{\geq}{\geqslant}

\def\ec{{\mathbb E}}
\def\ecL{{\mathbb E}} 

\def\CC{{\mathbb C}}
\def\RR{{\mathbb R}}
\def\NN{{\mathbb N}}
\def\ZZ{{\mathbb Z}}

\def\Dom{{\mathcal D}}

\def\({\left(}
\def\){\right)}
\def\<{\left\langle}
\def\>{\right\rangle}

\newcommand{\dd}{\mathrm{d}}

\newcommand{\supp}{\mathrm{supp}}

\newcommand{\tr}{\mathrm{tr}}

\newcommand{\be}{\begin{equation}}
\newcommand{\ee}{\end{equation}}
\newcommand{\bea}{\begin{eqnarray}}
\newcommand{\eea}{\end{eqnarray}}
\newcommand{\bee}{\begin{eqnarray*}}
\newcommand{\eee}{\end{eqnarray*}}

\begin{document}
\title{Regularity of the density of states for random Dirac operators}
\author{Sylvain Zalczer}\email{szalczer@bcamath.org}
\affiliation{BCAM - Basque Center for Applied Mathematics\\ Mazarredo, 14. 48009 Bilbao, Spain}

\date{\today}
\begin{abstract}
We consider the random Dirac operators for which we have proved Anderson localization in~[J.-M. Barbaroux, H.D. Cornean, and S. Zalczer. Localization for Gapped Dirac
Hamiltonians with Random Perturbations: Application to Graphene Antidot
Lattices. \emph{Doc. Math.}, 24:65–93, 2019]. We use the Wegner estimate we have got in that paper to prove Lipschitz regularity of the density of states. We use a method based on the Helffer-Sjöstrand formula to get information on the density of states from the Wegner estimate we have proven in the first article.
\end{abstract}
\maketitle

This article may be downloaded for personal use only. Any other use requires prior permission of the author and AIP Publishing. This article appeared in \emph{J. Math. Phys.} 64, 062103 (2023) and may be found at 
\href{https://doi.org/10.1063/5.0078383}{
https://doi.org/10.1063/5.0078383
}.

\section{Introduction}
The density of states is a major concept in condensed matter physics. It aims at measuring the number of energy levels by unit of volume near a given energy.
Since operators used in solid-state physics typically have  absolutely continuous or dense pure point spectrum, it is not possible to count the eigenvalues. 
Instead, we have to use a finite volume operator, the spectrum of which  is discrete at least in some region of interest. More information about the general theory of the density of states for random Schrödinger operators can be found in \onlinecite{KM} and \onlinecite{veselic}. Note that, in these references, all results are about \emph{Schrödinger} operators, \textit{i.e.} with the energy given by a Laplacian. The density of states of random \emph{Dirac} operators have been studied in a recent paper of Prado, de Oliveira and de Oliveira~\cite{PdOdO}; they consider the case of a discrete one-dimensional operator while, in the present paper, we consider a continuous model in arbitrary dimension.

In this paper, we consider the random Dirac-like operators for which we proved Anderson localization in the gap of the unperturbed operator in~\onlinecite{BCZ}. One important physical meaning of these operators is the study of graphene. Indeed, the dispersion surface of the Hamiltonian describing an electron in graphene has a conical singularity, at a point which is called \emph{Dirac point}\cite{CGPNGG}. Models with Dirac operators are used by physicists in numerical simulations and give results in accordance with the ones given by discrete models~\cite{BTP,Pedersen2,Pedersen1}.
Mathematically speaking, Fefferman and Weinstein prove that the Dirac equation governs the effective dynamics of electrons in graphene, when the initial data are spectrally concentrated near the singularity~\cite{FW2}. Note that the density of states of graphene models is studied by physicists (see~\onlinecite{AO} and references therein).

We look again at the model of \onlinecite{BCZ}. We assume that a periodic potential creates an energy bandgap, so that the graphene is a semiconductor~\cite{DOW}. We perturb this gapped Hamiltonian with a random potential, which represents the effect of impurities in the sample. We use the Wegner estimate proven in~\onlinecite{BCZ} to get the Lipschitz regularity of the density of states. As we will explain in Remark~\ref{theremark}, several standard methods used for the density of states cannot be directly adapted to the case of Dirac operators. We use  the Helffer-Sj\"{o}strand formula to get the result, as it  has been used in different settings  by Klopp and Raikov\cite{KR}, Germinet and Klein~\cite{GKcompre}, or Hislop~\cite{His}. We prove that this method can be adapted to the case of random Dirac operators, which implies to prove an adapted Combes-Thomas inequality (Property~\ref{ct-trace}). Such a result is usually called ``Lipschitz regularity of the integrated density of states'', the integrated density of states being the cumulative distribution function of the density of states measure. Nevertheless,  our operator is not bounded from below and, conversely to Schrödinger operators with Gaussian potentials~\cite{FLM} \cite{HLMW} \cite{FHLM}, finite -volume operators have infinitely many eigenvalues below any real number. It is thus not possible to define the cumulative distribution function in the usual way. Hence, our results will be formulated in terms of the density of states measure.

In Section~\ref{sec:2}, we present the model and give the result. In Section~\ref{sec:3} is given the proof of the theorem. In Section~\ref{sec:4}, we prove that another way to define the density of states gives the same value. In Appendix~\ref{sec:app}, we give a proof of a Combes-Thomas estimate.

\section{Model and important properties\label{sec:2}}

%
%

In this first part, we recall the model we have already used in \onlinecite{BCZ}. We use the terminology introduced by Klein and Koines in~\onlinecite{KleinKoines2004}.

\begin{definition}

Let $\{\sigma_i\}_{i=1}^d$ be a family of $n\times n$ Hermitian matrices where 
$n,d \geqslant 1$. 
We consider the following \emph{first-order linear operator} with constant coefficients: 
\begin{equation}\label{def:D0}                                     
 \sigma\cdot(-\iu\nabla) := \sum_{j=1}^d \sigma_j(-\iu\frac{\partial}{\partial x_j}),
\end{equation}
densely defined in $L^2(\RR^d, \CC^n)$. It is {\em elliptic} if there exists $C>0$ such that for all $p\in\RR^d$ and $q\in\CC^n$ we have 

\begin{align}\label{hc2}
\|(\sigma\cdot p) q\|_{\CC^n} \geqslant C\|p\|_{\RR^d}\;\|q\|_{\CC^n}.
\end{align}
\end{definition}
%
%
%
%
\begin{definition}
 We say that an operator on $L^2(\RR^d, \CC^n)$ is a {\em   coefficient positive operator} if it is a bounded invertible operator given by the multiplication by an $n\times n$ Hermitian matrix-valued measurable function
 $S(x)$ such that there exist two positive constants $S_\pm$ such that:
 \begin{equation}\label{def:S}
  0<S_-I_n \leqslant S(x) \leqslant S_+I_n,
 \end{equation}
where $I_n$ is the $n\times n$ identity matrix. 
\end{definition}

We consider operators of the type
\begin{equation}\label{hc1}
H_0:=S D_0 S+V_0
\end{equation}
where $D_0$ is a first-order elliptic operator with constant coefficients like in \eqref{def:D0}, and $S$ is a coefficient positive operator as in  \eqref{def:S}.
We assume that the function $S\in W^{1,\infty}(\RR^d, \mathcal{H}_n(\CC))$, where $\mathcal{H}_n$ is the space of $n\times n$ Hermitian matrices, is  $\ZZ^d$-periodic. We denote 
 $$
 D_S := S D_0 S.
 $$
The potential $V_0$ is $\ZZ^d$-periodic and belongs to $ L^\infty(\mathbb{R}^d, \mathcal{H}_n)$.  

With the above definitions and assumptions, the operator $H_0$ is self-adjoint on $H^1(\RR^d, \CC^n)$.

\begin{assumption}[gap assumption]\label{assump-gap} The spectrum of $H_0$ contains a finite open gap,  which will be denoted $(B_-, B_+)$.
\end{assumption}
Examples of operators satisfying these conditions can be found in~\onlinecite[pp. 68--69]{BCZ}. They include the free Dirac operator with a positive mass and the periodic operators modelling graphene antidot lattices studied in~\onlinecite{barbaroux}.

For operators fulfilling Assumption~\ref{assump-gap}, we want to study the effect of random perturbations on the  spectral gap $(B_-, B_+)$.  

The random matrix-valued perturbation $V_{\omega}$ describing local defects is defined by
\begin{equation*}
 V_{\omega} :=
 \sum_{i\in\mathbb{Z}^d}\lambda_{i}(\omega)u(\cdot-\xi_{i}(\omega)-i) ,
\end{equation*}
for some $u$, $\lambda_i$ and $\xi_i$ satisfying Assumption~\ref{assump:2} below.
The total Hamiltonian is thus
\begin{equation}\label{defop}
 H_{\omega} := H_0 + V_{\omega} .
\end{equation}
 
\begin{assumption}\label{assump:2}
  \noindent (i) The real-valued random variables 
  $\{\lambda_{i}(\omega),i\in\mathbb{Z}^d\}$ are independent and identically distributed. Their common distribution is absolutely continuous with respect to Lebesgue measure, with a density $h$ such that $\|h\|_{L^\infty}<\infty$. 
We assume that $\supp (h) = :[-m,M] \neq \{0\}$ for some finite non-negative $m$ and $M$. 
  
 \noindent (ii) The variables $\{\xi_i(\omega),\, i\in\mathbb{Z}^d\}$ are independent and identically distributed, and they are also independent from the $\lambda_j$'s.  They take values in $B_R$ with $0<R<\frac{1}{2}$, where $B_R$ is the ball in $\RR^d$ with radius $R$ and centered at the origin.

 \noindent (iii) The single-site matrix potential $u$ is compactly supported with $\mathrm{supp}(u)\subset [-2, 2 ]^d$. In addition, $u$ is assumed to be continuous almost everywhere,  with $u\in L^\infty(\mathbb{R}^d, \mathcal{H}_n^+)$, where $\mathcal{H}_n^+$ is the space of $n\times n$ non-negative Hermitian matrices.

\end{assumption}

As stated in~\onlinecite[Remark~2.6(ii)]{BCZ}, $H_\omega$ is an ergodic family of operators and then has an almost surely deterministic spectrum.

Throughout this article we use the sup norm in $\RR^d$ :
$|x|:=\max\{|x_i|:i=1, \ldots ,d \}$. By $\Lambda_L$ we denote the open box of side $L>0$ centered at $0$:
  $$
  \Lambda_L:=\{y\in\RR^d;|y|<\frac{L}{2}\}.
  $$
%
%

Given a box $\Lambda_L$, we define the localized operator
\begin{equation}\label{eq:1.4}
 H_{\omega,L}:=H_0+\sum_{i\in\Lambda_L\cap\ZZ^d}\lambda_i(\omega)u_i(\cdot -\xi_i(\omega))
 = H_0 + V_{\omega,L},
\end{equation}where we denote $u_i:=u(\cdot-i)$. This operator is a self-adjoint unbounded operator on $L^2(\RR^d,\CC^n)$.

We can then define $R_{\omega,L}(z):=(H_{\omega,L}-z)^{-1}$ the resolvent of 
$H_{\omega,L}$ and $E_{\omega,L}(\cdot)$ its spectral projection.

The two most important properties used in the proof are the following.
The first one is the Wegner estimate proven in \onlinecite[Theorem~4.2]{BCZ}.
\begin{property}[Wegner estimate]\label{thm:wegner}
Suppose Assumptions~\ref{assump-gap} and ~\ref{assump:2} hold. 
For any compact subinterval $J$ of $(B_-,B_+)$, there exists a constant $C_J$ such that for all $a<b\in J$ and all $L>0$
 \begin{equation*}
  \ecL \left ( \tr(E_{\omega,L} ([a,b])) \right )\leqslant C_{J}\; (b-a)\; L^d.
\end{equation*}
\end{property}

The second property is  similar to Lemma~4.6 of \onlinecite{BCZ}. Its proof is given in appendix.
Recall that, for a positive self-adjoint operator, \begin{equation}
\tr(A):=\sum_{n=0}^{+\infty}\langle e_i,Ae_i\rangle,                                                 \end{equation}
where $(e_i)_{i\in\NN}$ is any orthonormal basis. The value of the trace is independent of the orthonormal basis, see Theorem~VI.18 of Reed and Simon\cite{RS1}.
We say that an operator $A$ is trace class if \[\|A\|_1:=\tr(|A|)<\infty.\] We denote the ideal of trace class operators by $\mathcal{T}_1$.
\begin{property}[Combes-Thomas estimate]\label{ct-trace}
Fix $E_m$ and $Y>0$. 
For all $E$, $y$ such that $|E|\leq E_m$, $|y|\leq Y$ and any pair of  bounded functions $\chi_1$ and $\chi_2$ with $\|\chi_i\|_\infty\leqslant 1$ for $i=1,2$ and $\chi_1$ compactly supported, such that the distance between their supports is $a\geqslant10$, 
the operator $\chi_1 (H_\omega-E-\iu y)^{-1} \chi_2$ is trace class.

Furthermore, there exist two constants $D>0$ and $\alpha>0$ such that for all  $E$, $y$,  $\chi_1$, $\chi_2$ satisfying these hypotheses and all $\omega\in\Omega$, we have
  \begin{equation}\label{eq:ct-trace}
  \|\chi_1 (H_\omega-E- \iu y)^{-1}\chi_2\|_1\leqslant \frac{D}{|y|^{2d+1}}\;|{\rm supp}(\chi_1)|\;e^{-\alpha |y| a}.
  \end{equation}
\end{property}

Our goal being a result on the density of states, we have to define it.

For any nonnegative bounded measurable compactly supported function $\phi$ on the real line, $\omega\in\Omega$ and $L\in\NN$, we denote by $\chi_L$ the characteristic function of $\Lambda_L$ and define \[\nu_{\omega,L}(\phi):=\frac{1}{L^d}\tr\left(\chi_L\phi(H_\omega)\chi_L\right).\]

\begin{lemma}
For almost all $\omega$, all $L$ and $\phi$ satisfying the above conditions, $\nu_{\omega,L}(\phi)$ is finite. Moreover, we have that
 \[\lim_{L\to\infty} \nu_{\omega,L}(\phi)=\ec(\tr(\chi_1\phi(H_\omega)\chi_1)),\] where the convergence is both almost surely and  in $L^1(\Omega)$.
\end{lemma}
\begin{proof}
 We can write $\chi_L$ as $ \sum_{\gamma\in \ZZ^d\cap \Lambda_L} \chi_\gamma$ where $\chi_\gamma$ is the characteristic function of $\Lambda_1(\gamma)$, defined as $\{y\in\RR^d;|y-\gamma|<\frac{1}{2}\}$. 
 Then,
 \[\chi_L\phi(H_\omega)\chi_L=\sum_{\gamma,\gamma'\in \ZZ^d\cap\Lambda_L}\chi_{\gamma}\phi(H_\omega)\chi_{\gamma'}.\]

Let us prove that $\chi_\gamma\phi(H_{\omega})$ is trace class. 

Obviously, we have that for $\lambda>0$ \[\chi_\gamma\phi(H_{\omega})=\chi_\gamma(H_\omega+\iu\lambda)^{-2d}(H_\omega+\iu\lambda)^{2d}\phi(H_\omega).\]

Since $\phi$ is compactly supported, it is easy to see that for $\lambda>0$ the operator $(H_\omega+\iu\lambda)^{2d}\phi(H_\omega)$ is bounded.
 Let $T^{-1}$ be the multiplication by $\langle x\rangle^{-2d}$. We prove as in the proof of (SGEE) in \onlinecite[pp. 78-79]{BCZ} that for $\lambda$ large enough $T^{-1}(H_\omega+\iu\lambda)^{-2d}\in \mathcal{T}_1$ with a trace norm almost surely independent of $\omega$.
Hence, $\chi_\gamma\phi(H_\omega)\in\mathcal{T}_1$ and its trace norm is almost surely independent of $\omega$.

%
%
%
We can then use the cyclicity of the trace (\textit{cf.}~\onlinecite[Theorem~VI.25]{RS1}) to write 
 \[\nu_{\omega,L}(\phi)=\frac{1}{L^d}\sum_{\gamma\in \ZZ^d\cap \Lambda_L}\tr\left(\chi_\gamma\phi(H_\omega)\chi_\gamma\right).\]
 But ergodicity makes it possible to write 
 $\chi_\gamma\phi(H_\omega)\chi_\gamma=\chi_0 U_\gamma^*\phi(H_\omega)U_\gamma \chi_0=\chi_0\phi(H_{\tau_\gamma(\omega)})\chi_0$
where $U_\gamma$ is the translation operator associated with $\gamma$ on $L^2(\RR^d,\CC^n)$ and $\tau_\gamma$ the corresponding operator on $\Omega$. We have then:
\[\nu_{\omega,L}(\phi)=\frac{1}{L^d}\sum_{\gamma\in \ZZ^d\cap \Lambda_L}\tr\left(\chi_0\phi(H_{\tau_\gamma(\omega)})\chi_0\right).\]
 Since the family $\left(\tr\left(\chi_0\phi(H_{\tau_\gamma(\omega)})\chi_0\right)\right)_{\gamma\in\ZZ^d}$ is ergodic with respect to translations and because $\ec\left(|\tr\left(\chi_0\phi(H_{\tau_\gamma(\omega)})\chi_0\right)|\right)<\infty$, according to Birkhoff's ergodic theorem (\textit{cf.} for example~\onlinecite[ Theorem~1.14 and Corollary~1.14.1]{wa}), we get the result.
 \end{proof}

 We denote by $\nu(\phi)$ the almost sure limit of $\nu_{\omega,L}$ as $L$ tends to infinity, it can be defined for functions which are not nonnegative too.
 It is easy to see that $ \nu$ is a positive linear form on $ \mathcal{C}_c(\RR)$ and thus, through Riesz-Markov theorem, comes from a Borel measure. This measure will be called the \emph{density of states} and called $\nu$ too.

 Our result is then the following:
 \begin{theorem}[Lipschitz continuity]\label{thm:main}
For any compact subinterval $J\subset(B_-,B_+)$, there exists some constant $C_J>0$ such that for any $a,b\in J$, $a<b$, \[\nu([a,b])\leq C_J (b-a).\]
 \end{theorem}

\section{Proof of the main result\label{sec:3}}
 Since our Wegner estimate concerns $H_{\omega,L}$, we begin by showing that we can use this operator instead of $H_\omega$ in the density of states.

\begin{lemma}\label{HS}
Let $\phi\in\mathcal{C}^\infty_0(\RR)$ with support in $(B_-,B_+)$, as defined in Assumption~\ref{assump-gap}.
Then, uniformly in $\omega$,
\begin{equation}
\lim_{L\to\infty} \frac{1}{L^d}\left(\tr(\chi_L\phi(H_\omega)\chi_L)-\tr(\chi_L\phi(H_{\omega,L+15})\chi_L)\right)=0.
\end{equation}
\end{lemma}
\begin{proof}

According to the Helffer-Sjöstrand formula (\textit{cf.} for example \onlinecite[Proposition~4.8]{FK14},  or~\onlinecite[Section~2.2]{Davies}), there exists $\tilde{\phi}\in\mathcal{C}^\infty_c(\CC)$ which is an analytic extension of $\phi$,  such that
$\exists C>0,\forall z\in\CC, \left|\frac{\partial}{\partial\bar{z}}\tilde{\phi}(z)\right|\leq C|\Im (z)|^{2d+2}$. Moreover, we have for any self-adjoint operator $A$
\begin {equation}
 \phi(A)=\frac{1}{\pi}\int_{\RR^2} \bar{\partial}\tilde{\phi}(x+\iu y)(A-x-\iu y)^{-1}\dd x\dd y.
\end {equation}

Let $L>0$. According to this formula and resolvent equation, we have that for all $\omega\in\Omega$
\begin{align}
 &\left|\frac{1}{L^d}\left(\tr(\chi_L\phi(H_\omega)\chi_L)-\tr(\chi_L\phi(H_{\omega,L+15})\chi_L)\right)\right|\\
 &\leq \frac{1}{\pi L^d}\int_{\RR^2} |\bar{\partial}\tilde{\phi}(x+\iu y)|\left\|\chi_L(H_\omega-x-\iu y)^{-1}V^{ext}_{\omega,L+15}(H_{\omega,L+15}-x-\iu y)^{-1}\chi_L\right\|_1\dd x\dd y&\\
 &\leq\frac{1}{\pi L^d}\int_{\RR^2} |\bar{\partial}\tilde{\phi}(x+\iu y)|\left\|\chi_L(H_\omega-x-\iu y)^{-1}V^{ext}_{\omega,L+15}\right\|_1\left\|(H_{\omega,L+15}-x-\iu y)^{-1}\chi_L\right\|\dd x\dd y
\end{align}
where we have denoted $V_{\omega,L+15}^{ext}:=V_\omega-V_{\omega,L+15}=V_{\omega,L+15}^{ext}(1-\chi_ {L+10}).$

The operator norm in the integral is bounded by $|y|^{-1}$. The first factor $  |\bar{\partial}\tilde{\phi}(x+\iu y)|$ is bounded by $C|y|^{2d+2}$. The product is hence bounded by some constant independent of $L$, $x$ and $y$ multiplied by $|y|^{2d+1}$.

Let us consider the trace-norm factor.
We have that 
\begin{align}
 \|\chi_L(H_\omega-x-\iu y)^{-1}V^{ext}_{\omega,L+15}\|_1\leq M_\infty\Big(&\left\|\chi_L(H_\omega-x-\iu y)^{-1}(\chi_{L+\sqrt{L}}-\chi_{L+10})\right\|_1\label{tr1}\\&+\left\|\chi_L(H_\omega-x-\iu y)^{-1}(1-\chi_{L+\sqrt{L}})\right\|_1\Big)\label{tr2}
\end{align}
where $M_\infty$ denotes the supremum of $|V_\omega|$.

Thanks to~\eqref{eq:ct-trace}, \eqref{tr1} is bounded by $$\frac{C}{|y|^{2d+1}}|\supp(\chi_{L+\sqrt{L}}-\chi_{L+12})|\leq \frac{C'}{|y|^{2d+1}}L^{d-1/2}$$ and~\eqref{tr2} by $$\frac{C}{|y|^{2d+1}}|\supp(\chi_{L})|e^{-\alpha |y|\sqrt{L}}= \frac{C}{|y|^{2d+1}}L^de^{-\alpha |y|\sqrt{L}}.$$
Since there exist a compact set $K$ and $M>0$ independent of $L$ such that $\supp(\tilde{\phi})\subset K \times [-M,M]$, we can bound
\[\int_K\int_{-M}^M C e^{-|y|\sqrt{L}}\dd y\dd x\]  by $C'\frac{1}{\sqrt{L}}(1-e^{-M\sqrt{L}})$ for some constant $C'$.
We have then that there exists $C$ such that for all $L>0$
\[\left|\frac{1}{L^d}\big(\tr(\chi_L\phi(H_\omega)\chi_L)-\tr(\chi_L\phi(H_{\omega,L})\chi_L)\big)\right|\leq \frac{C}{\sqrt{L}}\] which tends to 0 as $L$ tends to infinity.
\end{proof}

\begin{remark}\label{theremark}
 Several methods to make appear the finite volume operator in the calculation, described in \onlinecite{KM} and \onlinecite{veselic}, are specific to Schrödinger operators: one of them uses the Feynman-Kac formula and the other one Dirichlet-Neumann bracketing, which is not possible in our case because the Dirac operator is not self-adjoint neither with Dirichlet nor Neumann boundary conditions.  
 Another method, close to ours, is presented in~\onlinecite{kirsch} but works only for discrete models since it uses the fact that $(H_\omega-z)^{-1}\chi_L$ is trace-class for all $z\in\CC\backslash\RR$ and $L>0$, which is not the case on the continuum.
 
 Conversely, the method we present here, which had already been used for generalized Anderson Hamiltonians~\cite{GKcompre} or perturbed Landau Hamiltonians~\cite{KR}, works for a large class of operators.
 \end{remark}

\begin{proof}[Proof of Theorem~\ref{thm:main}]
Let $J$ be a compact subinterval of $(B_-,B_+)$. We want to prove that there exists some constant $C_J>0$ such that for all $a,b\in J$, $a<b$,
$$
 \nu(1_{[a,b]})\leq C_J(b-a) .
$$
Let $\phi$ be as in Lemma~\ref{HS} with $0\leqslant \phi\leq1_{[a,b]}$.
Then, for $L\in\NN$ and almost every $\omega$,
\begin{align*}
 \tr \left(\chi_L\phi(H_{\omega,L+15})\chi_L\right)&
 \leq\tr \left(\chi_L E_{\omega,L+15}([a,b])\chi_L \right)\\
  &\leq\tr \left(E_{\omega,L+15}([a,b])\right).
\end{align*}

Then, the Wegner estimate (Property~\ref{thm:wegner}) gives us that there exists $C_J$ such that, for all $\phi$, $a,b\in J$ and $L$  large enough, we have
\begin{align*}\ecL \left( \tr\left(\chi_L\phi(H_{\omega,L+15})\chi_L \right) \right)&\leq\ecL \left ( \tr(E_{\omega,L+15} ([a,b])) \right )\leqslant C_{J}\; (b-a)\; (L+15)^d\\&\leqslant C'_{J}\; (b-a)\; L^d.\end{align*}

We have that 
\begin{align*}
 \nu(\phi)&=\lim_{L\to\infty}\ec\left(\frac{1}{L^d}\tr\left(\chi_L\phi(H_\omega)\chi_L\right)\right)\\
 &=\lim_{L\to\infty}\frac{1}{L^d}\ec\left(\tr\left(\chi_L\phi(H_{\omega,L+15})\chi_L\right)\right).\\
\end{align*}
 But we know that for all $L$
 \[\frac{1}{L^d}\ec\left(\tr\left(\chi_L\phi(H_{\omega,L+15})\chi_L\right)\right)\leq C_J (b-a)\]
 so $\nu(\phi)\leq C_J(b-a).$
 
 By taking the supremum on $0\leqslant \phi\leq1_{[a,b]}$ through monotone convergence, we get 
  the result on the density of states. 
\end{proof}

\section{Alternative definition\label{sec:4}}
Another way to define the density of states is to spatially cut off the operator before taking a function of it. Typically, if we are interested in the density of states of an operator $H$, we look at limits such as
\[\lim_{L\to\infty}\frac{1}{L^d}\tr(\phi(H_L))\]
where $H_L$ is the operator $H$ restricted to $\Lambda_L$ with good boundary conditions.

In our case,  we will denote by $H_{\omega,L}^{per}$  the operator $H_{\omega}$ restricted to the box $\Lambda_L$ with periodic boundary condition.

We begin by proving a technical lemma similar to the one of Birman and Solomyak given in \onlinecite[Theorem~4.1]{simon}. 
Recall that a self-adjoint operator is in the Schatten class $\mathcal{T}_p$ if its $p$th Schatten norm is finite, namely
\[\|A\|_p:=\left(\tr(|A|^p)\right)^{1/p}<\infty.\]

 \begin{lemma}\label{BS}
  Let $L>0$ and $\nabla_L^{per}$ the gradient with periodic boundary conditions on $\Lambda_L$. Let $H=f(x)g(-\iu\nabla_L^{per})$ on $L^2(\Lambda_L)$ with $f\in L^p(\Lambda_L)$ and $\sum_{n\in\ZZ^d}|g(\frac{2n\pi}{L})|^p<\infty$ for some $2\leq p<\infty$. Then, $H$ is in $\mathcal{T}_p$ and 
  \begin{equation}
   \|H\|_p\leq \frac{1}{L^{d/p}}\|f\|_p\left(\sum_{n\in\ZZ^d}|g(\frac{2n\pi}{L})|^p\right)^{1/p}.
  \end{equation}

 \end{lemma}
\begin{proof} Let us begin by the case $p=2$.
 For $\psi\in L^2(\Lambda_L)$, we recall that we can define for all $n\in\ZZ^d$ its Fourier coefficient by
 \[
  c_n(\psi):=\frac{1}{L^d}\int_{\Lambda_L}\psi(x)e^{-\frac{2\iu\pi n\cdot x}{L}}\dd x.
 \]

 Standard manipulations on Fourier series make it possible to prove that $H$ is an integral operator with kernel $\frac{1}{L^d}f(x)\check{g}(x-y)$ where $\check{g}$ is the function on $\Lambda_L$ of which the Fourier coefficients are $(g(\frac{2n\pi}{L}))_{n\in\ZZ^d}$. Then, according to Theorem~2.11 of \onlinecite{simon}, the operator is Hilbert-Schmidt with a norm equal to the $L^2$ norm of its kernel.
 
 Parseval identity gives the value of this norm.
 
 The general case comes from interpolation as in Theorem~4.1 of \onlinecite{simon}.
\end{proof}

\begin{lemma}\label{reduc}
Let $\phi$ be a bounded, measurable function with compact support. Uniformly in $\omega$, we have:
 \[\lim_{L\to\infty} \frac{1}{L^d}\tr(\phi(H_{\omega,L}^{per})-\chi_{L-10}\phi(H_{\omega,L}^{per})\chi_{L-10})=0.\]
\end{lemma}
\begin{proof}
 We can easily see that 
 \[\phi(H_{\omega,L}^{per})-\chi_{L-10}\phi(H_{\omega,L}^{per})\chi_{L-10}=\phi(H_{\omega,L}^{per})(\chi_L-\chi_{L-10})+(\chi_L-\chi_{L-10})\phi(H_{\omega,L}^{per})\chi_{L-10}.\]
 
 Moreover, $\phi(H_{\omega,L}^{per})(\chi_L-\chi_{L-10})=\phi(H_{\omega,L}^{per})(H_{\omega,L}^{per}-\iu)^{2d}(H_{\omega,L}^{per}-\iu)^{-2d}(\chi_L-\chi_{L-10})$. The factor $\phi(H_{\omega,L}^{per})(H_{\omega,L}^{per}-i)^{2d}$ is bounded with a bound depending only on $\supp(\phi)$ and $\|\phi\|_\infty$.
 
 According to Lemma~\ref{BS}, $(\sigma\cdot(-\iu\nabla_L^{per})-\iu)^{-1}=(\sigma\cdot(-\iu\nabla_L^{per})-\iu)^{-1}\chi_L$ is in $\mathcal{T}_{2d}$ with a $2d$-norm which is smaller than $\|\chi_L\|_{2d}\left(\frac{1}{L^d}\sum_{n\in\ZZ^d}|g(\frac{2\pi n}{L})|^{2d}\right)^{1/2d}$ where $g(p):=(\sigma\cdot p-\iu)^{-1}$. It is easy to see that $\|\chi_L\|_{2d}=\sqrt{L}$ and the second factor tends to $\|g\|_{2d}$ which is independent of $L$. Then, there exists a constant $C$ such that the product is smaller than $C\sqrt{L}$. Writing the resolvent equation as in Remark~2.6 (iii) of \onlinecite{BCZ}, we get the same result for $(H_{\omega,L}^{per}-\iu)^{-1}$. Then, by H\"older inequality,  $(H_{\omega,L}^{per}-\iu)^{-2d+1}$ is in $\mathcal{T}_{2d/(2d-1)}$ with norm $C L^{(2d-1)/2}$.
 
 Similarly, we prove that $(\chi_L-\chi_{L-10})(H_{\omega,L}^{per}-\iu)^{-1}\in\mathcal{T}_{2d}$ with a $2d$-norm $C\|\chi_L-\chi_{L-10}\|_{2d}=CL^\frac{d-1}{2d}$. By H\"older inequality, $(\chi_L-\chi_{L-10})(H_{\omega,L}^{per}-\iu)^{-2d}$ is trace-class and 
\[\|(\chi_L-\chi_{L-10})(H_{\omega,L}^{per}-\iu)^{-2d}\|_1\leq CL^{d-\frac{1}{2d}}. 
\]
\emph{A fortiori},  the operator $(\chi_L-\chi_{L-10})\phi(H_{\omega,L}^{per})\chi_{L-10}$ is trace-class too and its norm is bounded by $CL^{d-\frac{1}{2d}}$.
This completes the proof.
\end{proof}

For $L>0$, we will denote by $R_{\omega,L'}^{per}(E)$ the resolvent $(H_{\omega,L'}^{per}-E)^{-1}$. Furthermore, we define, for two operators $A$ and $B$, the commutator $[A;B]:=AB-BA$.
\begin{lemma}[Geometric resolvent equation]
 Let $L<L'\leq\infty$ and $\tilde{\chi}_L$ a smooth function with support in $\Lambda_L$. Then, for all $E\in\rho(H_{\omega,L'}^{per})\cap\rho(H_{\omega,L}^{per})$,
 \begin{equation}\label{GRE}
 \tilde{\chi}_LR_{\omega,L'}^{per}(E)=R_{\omega,L}^{per}(E)\tilde{\chi}_L+R_{\omega,L}^{per}(E)[H_{\omega,L'}^{per};\tilde{\chi_L}]R_{\omega,L'}^{per}(E).  
 \end{equation}\end{lemma}
\begin{proof}
 Let $\psi\in L^2(\Lambda_{L'})$. Then, $R_{\omega,L'}^{per}(E)\psi\in \Dom(H_{\omega,L'}^{per})$.
 Moreover, $\tilde{\chi}_LR_{\omega,L'}^{per}(E)\psi\in \Dom(H_{\omega,L'}^{per})$ and so
 \begin{equation}\label{GREinter}(H_{\omega,L'}^{per}-E)\tilde{\chi_L}R_{\omega,L'}^{per}(E)\psi=\tilde{\chi_L}\psi+[H_{\omega,L'}^{per};\tilde{\chi_L}]R_{\omega,L'}^{per}(E)\psi.\end{equation}
 We see that $\tilde{\chi_L}R_{\omega,L'}^{per}(E)\psi$ has support in $\Lambda_L$ and, as a function in $L^2(\Lambda_L)$, it is in $\Dom(H_{\omega,L}^{per})$. Since $[H_{\omega,L'}^{per};\tilde{\chi_L}]$ has support in $\Lambda_L$, we can project \eqref{GREinter} on $L^2(\Lambda_L)$ to get
\begin{equation}
 (H_{\omega,L}^{per}-E)\tilde{\chi_L}R_{\omega,L'}^{per}(E)\psi=\tilde{\chi_L}\psi+[H_{\omega,L'}^{per};\tilde{\chi_L}]R_{\omega,L'}^{per}(E)\psi.
\end{equation}
The desired result is found by multiplying by $R_{\omega,L}^{per}(E)$.
\end{proof}

\begin{lemma}
Let $\phi\in\mathcal{C}^\infty_0(\RR)$ with support in $(B_-,B_+)$, as defined in Assumption~\ref{assump-gap}.
Then, uniformly in $\omega$,
\begin{equation}
\lim_{L\to\infty} \frac{1}{L^d}\left(\tr(\chi_L\phi(H_\omega)\chi_L)-\tr(\phi(H_{\omega,L+10}^{per}))\right)=0.
\end{equation}
\end{lemma}
\begin{proof}

We will use again the Helffer-Sjöstrand formula: there exists $\tilde{\phi}\in\mathcal{C}^\infty_0(\CC)$, its trace on $\RR$ being $\phi$,  such that
$\exists C>0,\forall z\in\CC, \left|\frac{\partial}{\partial\bar{z}}\tilde{\phi}(z)\right|\leq C|\Im (z)|^{2d+2}$. Moreover, we have for any self-adjoint operator $A$
\begin {equation}
 \phi(A)=\frac{1}{\pi}\int_{\RR^2} \bar{\partial}\tilde{\phi}(x+\iu y)(A-x-\iu y)^{-1}\dd x\dd y.
\end {equation}

Let $L>0$. Lemma~\ref{reduc} will enable us to replace $\phi(H_{\omega,L+10}^{per})$ by $\chi_{L}\phi(H_{L+10}^{per})\chi_{L}$. We introduce a function $\tilde{\chi}_{L+10}$ which is smooth and in support in $\Lambda_{L+10}$ and which satisfies $\tilde{\chi}_{L+10}\chi_L=\chi_L$ and $\|\nabla\tilde{\chi}_{L+10}\|_\infty\leq1$. According to Helffer-Sjöstrand formula and our last two lemmas, we have that for all $\omega\in\Omega$
\begin{align*}
 &\left|\frac{1}{L^d}\left(\tr(\chi_L\phi(H_\omega)\chi_L)-\tr(\chi_{L}\phi(H_{\omega,L+10}^{per})\chi_{L})\right)\right|\\
 &\leq \frac{1}{\pi L^d}\int_{\RR^2} |\bar{\partial}\tilde{\phi}(x+\iu y)|\left\|\chi_L(H_\omega-x-\iu y)^{-1}\chi_L-\chi_L(H_{\omega,L+10}-x-\iu y)^{-1}\chi_L\right\|_1\dd x\dd y&\\
 &\leq \frac{1}{\pi L^d}\int_{\RR^2} |\bar{\partial}\tilde{\phi}(x+\iu y)|\left\|\chi_L\tilde{\chi}_{L+10}(H_\omega-x-\iu y)^{-1}\chi_L-\chi_L(H_{\omega,L+10}^{per}-x-\iu y)^{-1}\chi_L\right\|_1\dd x\dd y&\\
 &\leq \frac{1}{\pi L^d}\int_{\RR^2} |\bar{\partial}\tilde{\phi}(x+\iu y)|\Big\|\chi_L(H_{\omega,L+10}^{per}-x-\iu y)^{-1}\tilde{\chi}_{L+10}\chi_L\\&\hspace{4cm}+\chi_{L}(H_{\omega,L+10}^{per}-x-\iu y)^{-1}[H_\omega;\tilde{\chi}_{L+10}](H_\omega-x-\iu y)^{-1}\chi_L\\&\hspace{4cm}-\chi_L(H_{\omega,L+10}^{per}-x-\iu y)^{-1}\chi_L\Big\|_1\dd x\dd y&\\
 &\leq \frac{1}{\pi L^d}\int_{\RR^2} |\bar{\partial}\tilde{\phi}(x+\iu y)|\left\|\chi_{L}(H_{\omega,L+10}^{per}-x-\iu y)^{-1}[H_\omega;\tilde{\chi}_{L+10}](H_\omega-x-\iu y)^{-1}\chi_L\right\|_1\dd x\dd y&\\
 &\leq \frac{1}{\pi L^d}\int_{\RR^2} |\bar{\partial}\tilde{\phi}(x+\iu y)|\left\|\chi_{L}(H_{\omega,L+10}^{per}-x-\iu y)^{-1}\right\|\left\|[H_\omega;\tilde{\chi}_{L+10}](H_\omega-x-\iu y)^{-1}\chi_L\right\|_1\dd x\dd y.&
\end{align*}
We go from the penultimate to the last line by observing that $\tilde{\chi}_{L+10}\chi_L=\chi_L$.

The first factor is bounded by $C|y|^{2d+2}$ and the operator norm by $1/|y|$.

Let us consider the trace-norm factor.
Since $[H_\omega;\chi_{L+10}]=\sigma\cdot(-\iu\nabla\tilde{\chi}_{L+10})$ is a function with supremum 1 and support in the belt $\Lambda_{L+10}\backslash \Lambda_{L+5}$, equation \eqref{eq:ct-trace} makes it possible to bound the trace norm of $[H_\omega;\tilde{\chi}_{L+10}](H_\omega-x-\iu y)^{-1}\chi_L$ by $\frac{C}{|y|^{2d+1}}L^{d-1}$ with $C$ independent of $L$ and $(x,y)$.
We have then that there exists $C$ such that for all $L>0$
\[\frac{1}{L^d}\left(\tr(\chi_L\phi(H_\omega)\chi_L)-\tr(\chi_L\phi(H_{\omega,L+10}^{per})\chi_L)\right)\leq \frac{C}{L}\] which tends to 0 as $L$ tends to infinity.
\end{proof}
 
 \begin{acknowledgments}
 The author thanks the CPT in the university of Toulon where most of this work has been realized. He is very grateful  to H.D. Cornean for having suggested to use the Helffer-Sjöstrand formula. This work is supported the Basque Government through the BERC 2022-2025 program and by the Ministry of Science and Innovation: PID2020-112948GB-I00 funded by MCIN/AEI/10.13039/501100011033 and by "ERDF A way of making Europe". 
 \end{acknowledgments}
 
 \section*{Conflict of Interest and Data Availability Statements}
 The author  has  no conflicts to disclose.
 
 Data sharing is not applicable to this article as no new data were created or
analyzed in this study.
  \appendix
 \section{Combes-Thomas estimate\label{sec:app}} 
\begin{proof}[Proof of Property~\ref{ct-trace}]
Since the proof is similar to the one with real energies in \onlinecite{BCZ}, we give only what changes.

The first step is to bound the $\mathcal{T}_ {2d}$ norm of the product of the resolvent and a sufficiently decaying function.
According to inequality \eqref{hc2}, we have that for $E$ and $y$ in the specified intervals, $p\in \RR^d$ and $q\in \CC^n$,
\begin{equation}
 \|(\sigma\cdot p-E-\iu y)q\|^2\geq \|(\sigma\cdot p-E)q\|^2+|y|^2\|q\|^2\geq (|y|^2+\max(C\|p\|-|E|,0)^2)\|q\|^2.
\end{equation}

Then,
\begin{align}
\left\|\int_{\RR^d}(\sigma\cdot p-E-\iu y)^{-2d}\dd p\right\|&\leq \int_{\RR^d}\frac{1}{\sqrt{|y|^2+\max(C\|p\|-|E|,0)^2}^{2d}}\dd p\\
&\leq \int_{\|p\|\leq 2E_m/C}\frac{1}{|y|^{2d}}\dd p+\int_{\|p\|\geq 2E_m/C}\frac{2^{2d}}{C\|p\|^{2d}}\dd p\\&\leq C'+\frac{C''}{|y|^{2d}}
\end{align}

Using  the result of Birman and Solomyak (\onlinecite[Theorem 4.1]{simon}) and resolvent identities as in  \onlinecite[Remark~2.6~(iii)]{BCZ}, we can prove that, for $\chi$ a compactly supported function on $\RR^d$, 
we have for all $\omega$, $E$, $y$:
\begin{equation}
 \|R_\omega(E+\iu y)\chi\|_ {2d}\leq  \frac{C}{|y|}\-\| \chi\|_ {L^{2d}}
\end{equation}

for some constant $C$ depending only on $E_m$, $Y$, $C$ and $d$.



 We now prove an estimate in operator norm. 
 Let $\epsilon>0$ and define $\langle x-x_0\rangle_\epsilon := \sqrt{\epsilon+|x-x_0|^2}$. As in \onlinecite[Lemma~B.1]{BCZ}, for $t>0$, we define on $\mathcal{C}^\infty_c (\mathbb{R}^d,\mathbb{C}^n)$ the (non self-adjoint) operator
 $$
 H_{t,\epsilon} := e^{-t\langle x-x_0\rangle_\epsilon }
 H e^{t\langle x-x_0\rangle_\epsilon } 
 =H-tS\sigma\cdot(\iu\nabla\langle x-x_0\rangle_\epsilon)S.
 $$ Then, for all $t,\epsilon>0$, $\psi$ with norm 1, $E$ and $y \in \RR$, we have
 \begin{align*}
  \| (H_{t,\epsilon}-E-\iu y)\psi \| \geqslant & 
  \left|\Im(\langle \psi,(H_{t,\epsilon}-E-\iu y)\psi\rangle )\right|\\
  =  &\left| \langle \psi,tS\sigma\cdot(\nabla\langle x-x_0\rangle_\epsilon)S+y)\psi\rangle\right| \\ 
  \geqslant & |y| - \left| \langle \psi,tS\sigma\cdot(\nabla\langle x-x_0\rangle_\epsilon)S)\psi\rangle \right|
  \end{align*}

Choosing $t$ so small that $\|tS\sigma\cdot(\nabla\langle x-x_0\rangle_\epsilon)S)\|< \frac{|y|}{2}$, we get that the  norm is higher than $\frac{|y|}{2}$.
Hence, by a proof similar to Lemma~B.1 of \onlinecite{BCZ}, we get
\begin{equation}\label{BCapp}
 \displaystyle\| \chi_1 (H_\omega - E-\iu y)^{-1}\chi_2\| \leqslant
 \displaystyle\frac{2}{|y|} e^{\displaystyle- c|y|(a_2-a_1)}.
\end{equation}

To get the estimate in trace norm, we only have to follow the proof of Lemma~4.6 of \onlinecite{BCZ} which gives
\begin{equation}
  \|\chi_1 (H_\omega-E- \iu y)^{-1}\chi_2\|_1\leqslant \frac{D}{|y|^{2d+1}}\;|{\rm supp}(\chi_1)|\;e^{-\alpha |y| a}.
  \end{equation}
 \end{proof}

\bibliography{graphene}

\begin{thebibliography}{25}%
\makeatletter
\providecommand \@ifxundefined [1]{%
 \@ifx{#1\undefined}
}%
\providecommand \@ifnum [1]{%
 \ifnum #1\expandafter \@firstoftwo
 \else \expandafter \@secondoftwo
 \fi
}%
\providecommand \@ifx [1]{%
 \ifx #1\expandafter \@firstoftwo
 \else \expandafter \@secondoftwo
 \fi
}%
\providecommand \natexlab [1]{#1}%
\providecommand \enquote  [1]{``#1''}%
\providecommand \bibnamefont  [1]{#1}%
\providecommand \bibfnamefont [1]{#1}%
\providecommand \citenamefont [1]{#1}%
\providecommand \href@noop [0]{\@secondoftwo}%
\providecommand \href [0]{\begingroup \@sanitize@url \@href}%
\providecommand \@href[1]{\@@startlink{#1}\@@href}%
\providecommand \@@href[1]{\endgroup#1\@@endlink}%
\providecommand \@sanitize@url [0]{\catcode `\\12\catcode `\$12\catcode
  `\&12\catcode `\#12\catcode `\^12\catcode `\_12\catcode `\%12\relax}%
\providecommand \@@startlink[1]{}%
\providecommand \@@endlink[0]{}%
\providecommand \url  [0]{\begingroup\@sanitize@url \@url }%
\providecommand \@url [1]{\endgroup\@href {#1}{\urlprefix }}%
\providecommand \urlprefix  [0]{URL }%
\providecommand \Eprint [0]{\href }%
\providecommand \doibase [0]{http://dx.doi.org/}%
\providecommand \selectlanguage [0]{\@gobble}%
\providecommand \bibinfo  [0]{\@secondoftwo}%
\providecommand \bibfield  [0]{\@secondoftwo}%
\providecommand \translation [1]{[#1]}%
\providecommand \BibitemOpen [0]{}%
\providecommand \bibitemStop [0]{}%
\providecommand \bibitemNoStop [0]{.\EOS\space}%
\providecommand \EOS [0]{\spacefactor3000\relax}%
\providecommand \BibitemShut  [1]{\csname bibitem#1\endcsname}%
\let\auto@bib@innerbib\@empty
\bibitem [{\citenamefont {Kirsch}\ and\ \citenamefont {Metzger}(2007)}]{KM}%
  \BibitemOpen
  \bibfield  {author} {\bibinfo {author} {\bibfnamefont {W.}~\bibnamefont
  {Kirsch}}\ and\ \bibinfo {author} {\bibfnamefont {B.}~\bibnamefont
  {Metzger}},\ }\bibfield  {title} {\enquote {\bibinfo {title} {The
  {I}ntegrated {D}ensity of {S}tates for {R}andom {S}chrödinger
  {O}perators},}\ }\href {\doibase 10.1090/pspum/076.2/2307751} {\bibfield
  {journal} {\bibinfo  {journal} {in: F. Gesztesy, P. Deift, C. Galvez, P.
  Perry, W. Schlag (Editors): Spectral Theory and Mathematical Physics: A
  Festschrift in Honor of Barry Simon’s 60th Birthday.}\ ,\ \bibinfo {pages}
  {649–696}} (\bibinfo {year} {2007})}\BibitemShut {NoStop}%
\bibitem [{\citenamefont {Veseli\'c}(2008)}]{veselic}%
  \BibitemOpen
  \bibfield  {author} {\bibinfo {author} {\bibfnamefont {I.}~\bibnamefont
  {Veseli\'c}},\ }\href@noop {} {\emph {\bibinfo {title} {Existence and
  Regularity Properties of the Integrated Density of States of Random
  Schrödinger Operators}}},\ \bibinfo {series} {Lecture Notes in Mathematics},
  Vol.\ \bibinfo {volume} {1917}\ (\bibinfo  {publisher} {Springer, Berlin,
  Heidelberg},\ \bibinfo {year} {2008})\BibitemShut {NoStop}%
\bibitem [{\citenamefont {Prado}, \citenamefont {de~{O}liveira},\ and\
  \citenamefont {de~{O}liveira}(2021)}]{PdOdO}%
  \BibitemOpen
  \bibfield  {author} {\bibinfo {author} {\bibfnamefont {R.}~\bibnamefont
  {Prado}}, \bibinfo {author} {\bibfnamefont {C.}~\bibnamefont
  {de~{O}liveira}}, \ and\ \bibinfo {author} {\bibfnamefont {E.}~\bibnamefont
  {de~{O}liveira}},\ }\bibfield  {title} {\enquote {\bibinfo {title} {Density
  of states and {L}ifshitz tails for discrete 1{D} random {D}irac operators},}\
  }\href {\doibase 10.1007/s11040-021-09403-4} {\bibfield  {journal} {\bibinfo
  {journal} {Math. Phys. Anal. Geom.}\ }\textbf {\bibinfo {volume} {24}},\
  \bibinfo {pages} {30} (\bibinfo {year} {2021})}\BibitemShut {NoStop}%
\bibitem [{\citenamefont {Barbaroux}, \citenamefont {Cornean},\ and\
  \citenamefont {Zalczer}(2019)}]{BCZ}%
  \BibitemOpen
  \bibfield  {author} {\bibinfo {author} {\bibfnamefont {J.-M.}\ \bibnamefont
  {Barbaroux}}, \bibinfo {author} {\bibfnamefont {H.}~\bibnamefont {Cornean}},
  \ and\ \bibinfo {author} {\bibfnamefont {S.}~\bibnamefont {Zalczer}},\
  }\bibfield  {title} {\enquote {\bibinfo {title} {Localization for {G}apped
  {D}irac {H}amiltonians with {R}andom {P}erturbations: {A}pplication to
  {G}raphene {A}ntidot {L}attices},}\ }\href {\doibase
  10.25537/dm.2019v24.65-93} {\bibfield  {journal} {\bibinfo  {journal} {Doc.
  Math.}\ }\textbf {\bibinfo {volume} {24}},\ \bibinfo {pages} {65--93}
  (\bibinfo {year} {2019})}\BibitemShut {NoStop}%
\bibitem [{\citenamefont {Castro~Neto}\ \emph {et~al.}(2009)\citenamefont
  {Castro~Neto}, \citenamefont {Guinea}, \citenamefont {Peres}, \citenamefont
  {Novoselov},\ and\ \citenamefont {Geim}}]{CGPNGG}%
  \BibitemOpen
  \bibfield  {author} {\bibinfo {author} {\bibfnamefont {A.}~\bibnamefont
  {Castro~Neto}}, \bibinfo {author} {\bibfnamefont {F.}~\bibnamefont {Guinea}},
  \bibinfo {author} {\bibfnamefont {N.}~\bibnamefont {Peres}}, \bibinfo
  {author} {\bibfnamefont {K.}~\bibnamefont {Novoselov}}, \ and\ \bibinfo
  {author} {\bibfnamefont {A.}~\bibnamefont {Geim}},\ }\bibfield  {title}
  {\enquote {\bibinfo {title} {The electronic properties of graphene},}\ }\href
  {\doibase 10.1103/RevModPhys.81.109} {\bibfield  {journal} {\bibinfo
  {journal} {Rev. Mod. Phys.}\ }\textbf {\bibinfo {volume} {81}},\ \bibinfo
  {pages} {109--162} (\bibinfo {year} {2009})}\BibitemShut {NoStop}%
\bibitem [{\citenamefont {Brun}, \citenamefont {Thomsen},\ and\ \citenamefont
  {Pedersen}(2014)}]{BTP}%
  \BibitemOpen
  \bibfield  {author} {\bibinfo {author} {\bibfnamefont {S.}~\bibnamefont
  {Brun}}, \bibinfo {author} {\bibfnamefont {M.}~\bibnamefont {Thomsen}}, \
  and\ \bibinfo {author} {\bibfnamefont {T.}~\bibnamefont {Pedersen}},\
  }\bibfield  {title} {\enquote {\bibinfo {title} {Electronic and optical
  properties of graphene antidot lattices: comparison of {D}irac and
  tight-binding models},}\ }\href {\doibase 10.1088/0953-8984/26/26/265301}
  {\bibfield  {journal} {\bibinfo  {journal} {J. Phys. Condens. Matter}\
  }\textbf {\bibinfo {volume} {26}},\ \bibinfo {pages} {265301} (\bibinfo
  {year} {2014})}\BibitemShut {NoStop}%
\bibitem [{\citenamefont {Pedersen}\ and\ \citenamefont
  {Pedersen}(2012)}]{Pedersen2}%
  \BibitemOpen
  \bibfield  {author} {\bibinfo {author} {\bibfnamefont {J.}~\bibnamefont
  {Pedersen}}\ and\ \bibinfo {author} {\bibfnamefont {T.}~\bibnamefont
  {Pedersen}},\ }\bibfield  {title} {\enquote {\bibinfo {title} {Band gaps in
  graphene via periodic electrostatic gating},}\ }\href {\doibase
  10.1103/PhysRevB.85.235432} {\bibfield  {journal} {\bibinfo  {journal} {Phys.
  Rev. B}\ }\textbf {\bibinfo {volume} {85}},\ \bibinfo {pages} {235432}
  (\bibinfo {year} {2012})}\BibitemShut {NoStop}%
\bibitem [{\citenamefont {Pedersen}\ \emph {et~al.}(2012)\citenamefont
  {Pedersen}, \citenamefont {Gunst}, \citenamefont {Markussen},\ and\
  \citenamefont {Pedersen}}]{Pedersen1}%
  \BibitemOpen
  \bibfield  {author} {\bibinfo {author} {\bibfnamefont {J.}~\bibnamefont
  {Pedersen}}, \bibinfo {author} {\bibfnamefont {T.}~\bibnamefont {Gunst}},
  \bibinfo {author} {\bibfnamefont {T.}~\bibnamefont {Markussen}}, \ and\
  \bibinfo {author} {\bibfnamefont {T.}~\bibnamefont {Pedersen}},\ }\bibfield
  {title} {\enquote {\bibinfo {title} {Graphene antidot lattice waveguides},}\
  }\href {\doibase 10.1103/PhysRevB.86.245410} {\bibfield  {journal} {\bibinfo
  {journal} {Phys. Rev. B}\ }\textbf {\bibinfo {volume} {86}},\ \bibinfo
  {pages} {245410} (\bibinfo {year} {2012})}\BibitemShut {NoStop}%
\bibitem [{\citenamefont {Fefferman}\ and\ \citenamefont
  {Weinstein}(2014)}]{FW2}%
  \BibitemOpen
  \bibfield  {author} {\bibinfo {author} {\bibfnamefont {C.}~\bibnamefont
  {Fefferman}}\ and\ \bibinfo {author} {\bibfnamefont {M.}~\bibnamefont
  {Weinstein}},\ }\bibfield  {title} {\enquote {\bibinfo {title} {Wave packets
  in honeycomb structures and two-dimensional {D}irac equations},}\ }\href
  {\doibase 10.1007/s00220-013-1847-2} {\bibfield  {journal} {\bibinfo
  {journal} {Comm. Math. Phys.}\ }\textbf {\bibinfo {volume} {326}},\ \bibinfo
  {pages} {251--286} (\bibinfo {year} {2014})}\BibitemShut {NoStop}%
\bibitem [{\citenamefont {Ananyev}\ and\ \citenamefont
  {Ovchynnikov}(2017)}]{AO}%
  \BibitemOpen
  \bibfield  {author} {\bibinfo {author} {\bibfnamefont {V.}~\bibnamefont
  {Ananyev}}\ and\ \bibinfo {author} {\bibfnamefont {M.}~\bibnamefont
  {Ovchynnikov}},\ }\bibfield  {title} {\enquote {\bibinfo {title} {On the
  density of states of graphene in the nearest-neighbor approximation},}\
  }\href {\doibase 10.5488/CMP.20.43705} {\bibfield  {journal} {\bibinfo
  {journal} {Condens. Matter Phys.}\ }\textbf {\bibinfo {volume} {20}},\
  \bibinfo {pages} {43705} (\bibinfo {year} {2017})}\BibitemShut {NoStop}%
\bibitem [{\citenamefont {Dvorak}, \citenamefont {Oswald},\ and\ \citenamefont
  {Wu}(2013)}]{DOW}%
  \BibitemOpen
  \bibfield  {author} {\bibinfo {author} {\bibfnamefont {M.}~\bibnamefont
  {Dvorak}}, \bibinfo {author} {\bibfnamefont {W.}~\bibnamefont {Oswald}}, \
  and\ \bibinfo {author} {\bibfnamefont {Z.}~\bibnamefont {Wu}},\ }\bibfield
  {title} {\enquote {\bibinfo {title} {Bandgap opening by patterning
  graphene},}\ }\href {\doibase 10.1038/srep02289} {\bibfield  {journal}
  {\bibinfo  {journal} {Sci. Rep.}\ }\textbf {\bibinfo {volume} {3}},\ \bibinfo
  {pages} {2289} (\bibinfo {year} {2013})}\BibitemShut {NoStop}%
\bibitem [{\citenamefont {Klopp}\ and\ \citenamefont {Raikov}(2006)}]{KR}%
  \BibitemOpen
  \bibfield  {author} {\bibinfo {author} {\bibfnamefont {F.}~\bibnamefont
  {Klopp}}\ and\ \bibinfo {author} {\bibfnamefont {G.}~\bibnamefont {Raikov}},\
  }\bibfield  {title} {\enquote {\bibinfo {title} {Lifshitz tails in constant
  magnetic fields},}\ }\href {\doibase 10.1007/s00220-006-0059-4} {\bibfield
  {journal} {\bibinfo  {journal} {Comm. Math. Phys.}\ }\textbf {\bibinfo
  {volume} {267}},\ \bibinfo {pages} {669} (\bibinfo {year}
  {2006})}\BibitemShut {NoStop}%
\bibitem [{\citenamefont {Klein}\ and\ \citenamefont
  {Germinet}(2013)}]{GKcompre}%
  \BibitemOpen
  \bibfield  {author} {\bibinfo {author} {\bibfnamefont {A.}~\bibnamefont
  {Klein}}\ and\ \bibinfo {author} {\bibfnamefont {F.}~\bibnamefont
  {Germinet}},\ }\bibfield  {title} {\enquote {\bibinfo {title} {A
  comprehensive proof of localization for continuous {A}nderson models with
  singular random potentials},}\ }\href {\doibase 10.4171/JEMS/356} {\bibfield
  {journal} {\bibinfo  {journal} {J. Eur. Math. Soc.}\ }\textbf {\bibinfo
  {volume} {15}},\ \bibinfo {pages} {53--143} (\bibinfo {year}
  {2013})}\BibitemShut {NoStop}%
\bibitem [{\citenamefont {Hislop}(2008)}]{His}%
  \BibitemOpen
  \bibfield  {author} {\bibinfo {author} {\bibfnamefont {P.}~\bibnamefont
  {Hislop}},\ }\bibfield  {title} {\enquote {\bibinfo {title} {Lectures on
  random {S}chrödinger operators},}\ }\href {\doibase 10.1090/conm/476/09293}
  {\bibfield  {journal} {\bibinfo  {journal} {Contemp. Math.}\ }\textbf
  {\bibinfo {volume} {476}},\ \bibinfo {pages} {41--131} (\bibinfo {year}
  {2008})}\BibitemShut {NoStop}%
\bibitem [{\citenamefont {Fischer}, \citenamefont {Leschke},\ and\
  \citenamefont {Müller}(2000)}]{FLM}%
  \BibitemOpen
  \bibfield  {author} {\bibinfo {author} {\bibfnamefont {W.}~\bibnamefont
  {Fischer}}, \bibinfo {author} {\bibfnamefont {H.}~\bibnamefont {Leschke}}, \
  and\ \bibinfo {author} {\bibfnamefont {P.}~\bibnamefont {Müller}},\
  }\bibfield  {title} {\enquote {\bibinfo {title} {Spectral localization by
  {G}aussian random potentials in multi-dimensional continuous space.}}\ }\href
  {\doibase 10.1023/A:1026425621261} {\bibfield  {journal} {\bibinfo  {journal}
  {J. Stat. Phys.}\ }\textbf {\bibinfo {volume} {101}},\ \bibinfo {pages}
  {935--985} (\bibinfo {year} {2000})}\BibitemShut {NoStop}%
\bibitem [{\citenamefont {Hupfer}\ \emph {et~al.}(2001)\citenamefont {Hupfer},
  \citenamefont {Leschke}, \citenamefont {Müller},\ and\ \citenamefont
  {Warzel}}]{HLMW}%
  \BibitemOpen
  \bibfield  {author} {\bibinfo {author} {\bibfnamefont {T.}~\bibnamefont
  {Hupfer}}, \bibinfo {author} {\bibfnamefont {H.}~\bibnamefont {Leschke}},
  \bibinfo {author} {\bibfnamefont {P.}~\bibnamefont {Müller}}, \ and\
  \bibinfo {author} {\bibfnamefont {S.}~\bibnamefont {Warzel}},\ }\bibfield
  {title} {\enquote {\bibinfo {title} {The absolute continuity of the
  integrated density of states for magnetic {S}chrödinger operators with
  certain unbounded random potentials.}}\ }\href {\doibase
  10.1007/s002200100467} {\bibfield  {journal} {\bibinfo  {journal} {Comm.
  Math. Phys.}\ }\textbf {\bibinfo {volume} {221}},\ \bibinfo {pages}
  {229--254} (\bibinfo {year} {2001})}\BibitemShut {NoStop}%
\bibitem [{\citenamefont {Fischer}\ \emph {et~al.}(1997)\citenamefont
  {Fischer}, \citenamefont {Hupfer}, \citenamefont {Leschke},\ and\
  \citenamefont {Müller}}]{FHLM}%
  \BibitemOpen
  \bibfield  {author} {\bibinfo {author} {\bibfnamefont {W.}~\bibnamefont
  {Fischer}}, \bibinfo {author} {\bibfnamefont {T.}~\bibnamefont {Hupfer}},
  \bibinfo {author} {\bibfnamefont {H.}~\bibnamefont {Leschke}}, \ and\
  \bibinfo {author} {\bibfnamefont {P.}~\bibnamefont {Müller}},\ }\bibfield
  {title} {\enquote {\bibinfo {title} {Existence of the density of states for
  multi-dimensional continuum {S}chrödinger operators with {G}aussian random
  potentials},}\ }\href {\doibase 10.1007/s002200050236} {\bibfield  {journal}
  {\bibinfo  {journal} {Comm. Math. Phys.}\ }\textbf {\bibinfo {volume}
  {190}},\ \bibinfo {pages} {133--141} (\bibinfo {year} {1997})}\BibitemShut
  {NoStop}%
\bibitem [{\citenamefont {Klein}\ and\ \citenamefont
  {Koines}(2004)}]{KleinKoines2004}%
  \BibitemOpen
  \bibfield  {author} {\bibinfo {author} {\bibfnamefont {A.}~\bibnamefont
  {Klein}}\ and\ \bibinfo {author} {\bibfnamefont {A.}~\bibnamefont {Koines}},\
  }\bibfield  {title} {\enquote {\bibinfo {title} {A general framework for
  localization of classical waves. {II}. {R}andom media},}\ }\href {\doibase
  10.1023/B:MPAG.0000024653.29758.20} {\bibfield  {journal} {\bibinfo
  {journal} {Math. Phys. Anal. Geom.}\ }\textbf {\bibinfo {volume} {7}},\
  \bibinfo {pages} {151--185} (\bibinfo {year} {2004})}\BibitemShut {NoStop}%
\bibitem [{\citenamefont {Barbaroux}, \citenamefont {Cornean},\ and\
  \citenamefont {Stockmeyer}(2017)}]{barbaroux}%
  \BibitemOpen
  \bibfield  {author} {\bibinfo {author} {\bibfnamefont {J.-M.}\ \bibnamefont
  {Barbaroux}}, \bibinfo {author} {\bibfnamefont {H.}~\bibnamefont {Cornean}},
  \ and\ \bibinfo {author} {\bibfnamefont {E.}~\bibnamefont {Stockmeyer}},\
  }\bibfield  {title} {\enquote {\bibinfo {title} {Spectral gaps in graphene
  antidot lattices},}\ }\href {\doibase 10.1007/s00020-017-2411-9} {\bibfield
  {journal} {\bibinfo  {journal} {Integr. Equ. Oper. Theory}\ }\textbf
  {\bibinfo {volume} {89}},\ \bibinfo {pages} {631--646} (\bibinfo {year}
  {2017})}\BibitemShut {NoStop}%
\bibitem [{\citenamefont {Reed}\ and\ \citenamefont {Simon}(1981)}]{RS1}%
  \BibitemOpen
  \bibfield  {author} {\bibinfo {author} {\bibfnamefont {M.}~\bibnamefont
  {Reed}}\ and\ \bibinfo {author} {\bibfnamefont {B.}~\bibnamefont {Simon}},\
  }\href@noop {} {\emph {\bibinfo {title} {Methods of Modern Mathematical
  Physics}}},\ Vol.\ \bibinfo {volume} {I: Functional Analysis}\ (\bibinfo
  {publisher} {Elsevier Science},\ \bibinfo {year} {1981})\BibitemShut
  {NoStop}%
\bibitem [{\citenamefont {Walters}(2000)}]{wa}%
  \BibitemOpen
  \bibfield  {author} {\bibinfo {author} {\bibfnamefont {P.}~\bibnamefont
  {Walters}},\ }\href@noop {} {\emph {\bibinfo {title} {An introduction to
  ergodic theory}}},\ \bibinfo {series} {Graduate texts in Mathematics},
  Vol.~\bibinfo {volume} {79}\ (\bibinfo  {publisher} {Springer Science \&
  Business Media},\ \bibinfo {year} {2000})\BibitemShut {NoStop}%
\bibitem [{\citenamefont {Fermanian~Kammerer}(2014)}]{FK14}%
  \BibitemOpen
  \bibfield  {author} {\bibinfo {author} {\bibfnamefont {C.}~\bibnamefont
  {Fermanian~Kammerer}},\ }\enquote {\bibinfo {title} {Chaos en mécanique
  quantique},}\ \ (\bibinfo  {publisher} {Les \'{E}ditions de l'\'{E}cole
  polytechnique},\ \bibinfo {year} {2014})\ Chap.\ \bibinfo {chapter}
  {Opérateurs pseudo-différentiels semi-classiques}\BibitemShut {NoStop}%
\bibitem [{\citenamefont {Davies}(1995)}]{Davies}%
  \BibitemOpen
  \bibfield  {author} {\bibinfo {author} {\bibfnamefont {E.}~\bibnamefont
  {Davies}},\ }\href@noop {} {\emph {\bibinfo {title} {Spectral theory and
  differential operators}}}\ (\bibinfo  {publisher} {Cambridge University
  Press},\ \bibinfo {year} {1995})\BibitemShut {NoStop}%
\bibitem [{\citenamefont {Kirsch}(2008)}]{kirsch}%
  \BibitemOpen
  \bibfield  {author} {\bibinfo {author} {\bibfnamefont {W.}~\bibnamefont
  {Kirsch}},\ }\bibfield  {title} {\enquote {\bibinfo {title} {An invitation to
  random {S}chr\"odinger operators},}\ }in\ \href@noop {} {\emph {\bibinfo
  {booktitle} {Random {S}chr\"odinger Operators}}},\ \bibinfo {series}
  {Panoramas et Synth\`eses}, Vol.~\bibinfo {volume} {25}\ (\bibinfo
  {publisher} {{S}oci\'et\'e {M}ath\'ematique de {F}rance},\ \bibinfo {year}
  {2008})\BibitemShut {NoStop}%
\bibitem [{\citenamefont {Simon}(2010)}]{simon}%
  \BibitemOpen
  \bibfield  {author} {\bibinfo {author} {\bibfnamefont {B.}~\bibnamefont
  {Simon}},\ }\href@noop {} {\emph {\bibinfo {title} {Trace Ideals and Their
  Applications}}},\ Mathematical Surveys and Monographs\ (\bibinfo  {publisher}
  {American Mathematical Society},\ \bibinfo {year} {2010})\BibitemShut
  {NoStop}%
\end{thebibliography}%

\end{document}